\documentclass[letter]{article}
\usepackage[utf8]{inputenc}
\usepackage{amsmath, amssymb, amsthm}
\usepackage{nicefrac}

\usepackage{fullpage}

\usepackage{epsfig}
\usepackage{subfigure}
\usepackage{natbib}
\usepackage{algorithm,algorithmic}

\theoremstyle{plain}
\newtheorem{theorem}	 			{Theorem}
\newtheorem{lemma}		[theorem]	{Lemma}

\theoremstyle{definition}
\newtheorem{definition}[theorem]{Definition} 
\theoremstyle{remark} 
\newcommand{\E}{{\bf E}}
\newcommand{\opt}{\textsc{opt}}
\newcommand{\pmarker}{\textsc{PredictiveMarker}}

 \author{Michael Mitzenmacher\footnote{School of Engineering \& Applied Sciences, Harvard University, Cambridge, Massachusetts, USA.  \texttt{michaelm{@}eecs.harvard.edu}. \texttt{http://www.eecs.harvard.edu/\textasciitilde michaelm}. Supported by NSF grants CCF-1563710 and CCF-1535795.} \and Sergei Vassilvitskii\footnote{Google Research, New York, New York, USA. \texttt{sergeiv{@}google.com}. \texttt{http://theory.stanford.edu/\textasciitilde sergeiv}}}

\title{Algorithms with Predictions\footnote{This survey is to appear as a chapter in Beyond the Worst-Case Analysis of Algorithms, a collection edited by Tim Roughgarden.  We hope to occasionally update the survey here, with new versions that include discussions of new results and advances in the area of Algorithms with Predictions.  }
}\label{chap:awp}

\begin{document}

\maketitle

\abstract{We introduce algorithms that use predictions from machine learning applied to the 
 input to circumvent worst-case analysis. We aim for algorithms that have near optimal 
 performance when these predictions are good, but recover the prediction-less worst case behavior when the predictions have large errors. }

\section{Introduction}
In finding ways to go beyond worst case analysis, previous chapters have described different ways to model the inputs seen by an algorithm in order to avoid fragile bad examples, give better guarantees, or explain the efficacy of methods in practice.  Many of these approaches are based on assuming a model of the input that includes randomness in a very specific way. For instance, in average case analysis 
data is drawn from a fixed but unknown distribution, and with random arrival models 
the input is assumed to be randomly permuted.  In this chapter, instead of posing a specific model or a set of assumptions on the input, we provide a general framework designed to make use of the rapidly growing power of machine learning techniques.  In our framework, we assume that we have a machine learning method that provides us with predictions about the input, and we use the prediction to make a more effective algorithm.  We then analyze the performance of the algorithm as a function of how accurate the prediction is;  ideally, the better the prediction, the better the performance. 

One thing that distinguishes this approach from other work is its natural connection to practice, as for many problems machine learning can be readily applied to data to provide the necessary prediction for new inputs.  Moreover, if we can successfully tie the performance of an algorithm to the quality of the predictions it receives, then as machine learning technology evolves and the quality of predictions improves, we get better performing algorithms essentially for free. 

When designing these kinds of algorithms with predictions, there are several new challenges.
One is a new goal for our theoretical analysis.   We wish  to provide formal guarantees of the following form:  if our predictor has a given level of performance, our algorithm will achieve a corresponding level of performance.  
A further challenge is to identify what quantity or quantities to predict, as these will generally be problem specific.  Choosing the right quantity to predict can affect both the algorithm performance and the bounds from our analysis.  Finally, an additional challenge is that by nature machine learning methods are imperfect.  They have errors that can be large and surprising, and the algorithms we design using machine learning predictions should be robust  enough to cope with them. 



We start with some very simple examples suggesting why this framework might be useful, and then present some additional examples of more complicated algorithms and data structures that make use of predictions.  

\subsection{Warm-up: Binary search}
\label{sec:example}
As a first example, consider the binary search problem. Given a sorted array $A$ on $n$ elements and a query element $q$, the goal is to either find the index of $q$ in the array, or state that it is not in the set. The textbook method is binary search: compare the value of $q$ to the value of the middle element of $A$, and recurse on the correct half of the array. After $O(\log n)$ probes, the method either finds $q$ or correctly returns that $q$ is not in the array. 

Binary search optimizes for the worst case, but there are often times when we can do better. For example, most bookstores have books arranged alphabetically by the authors' last name within a particular section. If we were looking for an Agatha Christie mystery, we would likely start our search near the beginning of the section; if, instead, we were to look for a Dorothy Sayers novel, we'd start further towards the end. We first look at the approximate location where we expect to find the book, using our knowledge of the alphabet. 

How can we generalize this approach? Let us assume we have a predictor $h$ which, for every query, $q$,  returns our best guess for the position of $q$ in the array.  To use $h$, a natural approach is to first probe the location at $h(q)$; if $q$ is not found there, we immediately know whether it is smaller or larger. Suppose $q$ is larger than the element in $A[h(q)]$ and the array is sorted in increasing order. We probe elements at $h(q) + 2, h(q) + 4, h(q) + 8$, and so on, until we find an element larger than $q$ (or we hit the end of the array). Then we apply binary search on the interval that's guaranteed to contain $q$ (if it exists). The bookstore example uses interpolation search as a classifier; since 'C' is the third letter out of 26, we start our search for the Agatha Christie book about $\nicefrac{3}{26} \approx 12\%$ of the  way through the Mysteries section. 

What is the cost of such an approach, in terms of the number of comparisons? Let $t(q)$ be the true position of $q$ in the array (or the position of the largest element smaller than $q$ if it is not in the array). Suppose the error of the classifier on $q$ is $\eta_q = |h(q) - t(q)|$. The cost of running the above algorithm starting at $h(q)$ is at most $2(\log{ |h(q)-t(q)}|) = 2 \log \eta_q$. 

If the queries $q$ come from a distribution, then the expected cost of the algorithm is:
$$2 \mathbb{E}_q \Big[\log \left(|h(q) - t(q)|\right)\Big] \leq 2 \log \mathbb{E}_q \Big[|h(q) - t(q)|\Big] = 2 \log \mathbb{E}_q[\eta_q],$$
where the inequality follows by Jensen's inequality. This gives a guarantee on the performance of the algorithm parametrized by the error of the predictor.
In particular, even a classifier with an average error of $O(\text{poly}\!\log n)$ leads to an improvement in asymptotic performance. Moreover, since $\eta_q$ is trivially bounded by $n$, even an exceptionally bad predictor cannot do much harm.  

\subsection{Online algorithms: Ski Rental}
The above example has the nice property that the use of predictions is essentially free. On the one hand, as the prediction error tends to zero, the running time approaches the best possible for this task (a constant). On the other hand, the error is naturally bounded by the number of elements, so even bad predictions will not asymptotically degrade the algorithm's performance. In other situations there can be a more dramatic trade-off between the benefit of using the predictions and the cost incurred when these predictions are wildly incorrect. 

Consider the \textsc{SkiRental} problem. At the beginning of the ski season, a new skier has the option to buy skis for \$b dollars, or to rent them every day for \$1 per day. This is one of the simplest settings of decision making under uncertainty --- the skier does not know how many days she will ski, yet a simple deterministic strategy will guarantee that she does not spend more than twice as much as she would have had she known the future.

The algorithm achieving that bound rents skis for the first $b$ days, and then buys them on day $b+1$. If the skier skis $b$ or fewer days, she has spent the optimal amount. If, by chance, she stops skiing after day $b+1$, she's spent at most $\$2b$ in total, which is less than twice the optimal amount. 

Suppose the skier has access to a prediction $h(d)$ of how many days she will ski. How should she use this information? Let $d^*$ be the true number of skiing days, and $\eta = |h(d) - d^*|$ be the error in the prediction. It is easy to verify that the algorithm that treats the prediction as truth (i.e. buying skis on day 1 if $h(d) > b$ and renting daily otherwise) has a total cost of $OPT + \eta$. We observe that in this case, the use of predictions is not ``free.'' While the algorithm performs optimally when the prediction is correct, if the skier trusts the prediction and doesn't buy the skis when she should, she can spend arbitrarily more money than if she applies the simple deterministic strategy above. 

There is, however, a simple fix. Let $\lambda \in [0, 1]$ be a tunable parameter, and consider the following algorithm. If $h(d) > b$, the skier buys on day $\lceil \lambda b \rceil$, and otherwise, she buys on day $\lceil \nicefrac{ b}{\lambda} \rceil$. A case analysis shows that the competitive ratio of this algorithm is bounded by:
\begin{equation}
 1 + \min\Big(\frac{1}{\lambda}, \lambda + \frac{\eta}{(1 - \lambda) OPT }\Big).
\label{eqn:skirental}
 \end{equation}
In particular, as the error of the prediction drops to 0, the competitive ratio is no more than $1 + \lambda$. On the other hand, even for large errors, the ratio is never worse than $1 + \nicefrac{1}{\lambda}$.  Note that $\lambda = 1$ recovers the algorithm we described originally.  

\subsection{Model}
The two examples above outline the desiderata that we have for algorithms that use predictions. There are three properties that we highlight. 

First, we have isolated the inner workings of the predictor from the algorithm that uses the predictions, instead simply abstracting the predictor as a function $h$. Our algorithms are accordingly not tied to a specific type of predictor. We can apply decision trees, neural networks, or any other approach to obtain predictions; any $h$ with low error suffices. 

Second, the goal is to tie the performance of the algorithms to the observed loss of the predictor. In the setting of our examples, where we  used competitive analysis, we further isolated two concepts. We want the algorithms to be  {\em consistent}; 
that is, ideally their performance should recover that guaranteed by the offline optimal algorithm given an error-free prediction. Additionally, to capture the fact that machine learning systems sometimes have very large errors, we want algorithms to be {\em robust}; that is, ideally their performance should not be worse than standard online algorithms that use no predictions whatsoever. 

While ideal consistency and robustness may be quite challenging, we can loosen the goals using an approximation.  Formally, we say that an algorithm is {\em $\alpha$-consistent} if its competitive ratio tends to $\alpha$ as the error in the predictions goes to $0$, and {\em $\beta$-robust} if the competitive ratio is bounded by $\beta$ even with arbitrarily bad predictions. 

As we saw in the ski-rental example, there is often a tension between consistency and robustness. A practitioner who has high confidence in the predictions may aim for high consistency and low robustness by choosing a small value of $\lambda$. On the other hand, a risk-averse decision maker may choose a higher value for $\lambda$, limiting the benefit of the predictions but also the additional cost when they turn out to be incorrect. 

\section{Counting Sketches}

Another example of a problem where predictions have been shown to boost performance is in the setting of counting sketches for data streams.  We briefly describe the Count-Min sketch as an example of a counting sketch.  For simplicity, we assume items come in as a data stream one at a time;  for example, these could be URLs or IP addresses being accessed. Keeping a separate counter for each item may require too much space, and so we can instead use a sketch that requires less memory at the cost of obtaining only an approximate count for each item, with some chance of failure for each item.  The Count-Min sketch sets up a rectangular array of counters with $r$ rows and $c$ columns.  Each item hashes to one counter location in each row;  when an item passes in the stream, each of its counters is incremented.  The approximate count for an item is the minimum counter value associated with the item, which can only yield an overestimate of the actual count for the item.  Various results are known that show the error for such a sketch can be small for appropriate values of $r$ and $c$.  Note that if an item has at least one counter where no other item hashes to it, the resulting approximate count will in fact be the exact count. The idea behind the good performance of the Count-Min sketch is that for most items, there will be at least one counter for which the item collides with very few other items, leading to an accurate estimate.  In particular, for skewed data streams where item frequencies follow a Zipfian distribution (or more generally for heavy-tailed distributions), so much of the total count is based on a small number of items that the approach can be very accurate, as most collisions introduce only a small error in the counter.

Suppose, however, that we had a predictor that could reasonably accurately predict which items were the ``heavy hitters,'' that is the most frequent items.  Since the idea of using a data sketch is to save space, we do not want to use a separate counter for every item, but we may be willing to use space to keep individual counters for each item that is predicted to have a high count.  This assures accuracy for correctly predicted heavy hitters, which is often important, but also importantly it greatly reduces the possibility of a large error for an item with a small count, since removing potential heavy hitters from the larger array greatly reduces the possibility that a small item will have all of its counters collide with a large item.  

The works \citep{hsu2018learning,aamand2019learned} have formalized this high-level argument with provable results for the Count-Min sketch and Count-Sketch for Zipfian frequency distributions, showing they can improve the space/performance tradeoff over sketches without predictions.  They also show this improvement holds in practice.  While we do not go into further details here, the example of counting sketches provides an intuitive approach for using predictions within algorithms and data structures:  if there are a limited set of problematic elements, such as outliers or high weight elements, that greatly effect performance when they are not known in advance, a predictor may allow these elements to be separated out and correspondingly improve overall performance.     

\section{Learned Bloom Filters}

An early proposed example of how predictions from machine learning could improve data structures provided a novel variation of the Bloom filter \citep{kraska2018case}.  

To start, let us briefly review standard Bloom filters \citep{bloom1970space,broder2004network}, a data structure used to answer set membership queries using small space.  
A Bloom filter for representing a set $S = \{x_1,x_2,\ldots,x_n\}$ of $n$ elements corresponds to an array of $m$ bits, and uses $k$ independent hash functions $h_1,\ldots,h_k$ with range $\{0,\ldots,m-1\}$. 
Note that the number of bits per item used by the Bloom filter is given by $m/n$.
Here we assume that these hash functions are perfectly random.  
Initially all array bits are 0.  For each element $x \in S$, the array
bits $h_i(x)$ are set to 1 for $1 \leq i \leq k$; a bit may be repeatedly set to 1.  To check if an item $y$ is in
$S$, we check whether all $h_i(y)$ are set to 1.  If not, then clearly
$y$ is not a member of $S$.  If all $h_i(y)$ are set to 1, we conclude
that $y$ is in $S$, although this may be a {\em false positive}.  A
Bloom filter does not produce false negatives.

Let $y$ be an element such that $y \notin S$, where $y$ is chosen independently of the hash functions used to create the filter.
Let $\rho$ be the fraction of bits set to 1 after the elements are hashed. Then the probability of a false positive is $\rho^k$.   
Now the expected value of $\rho$ is easily calculated, as the probability a specific bit in the filter stays 0 is just
\begin{align*}
\left( 1 - \frac{1}{m} \right )^{kn} \approx e^{-kn/m}.
\end{align*}
Standard techniques show the $\rho$ is close to its expectation with high probability, so using the expectation in place of $\rho$, we see the false positive probability will be concentrated near
\begin{align*}
(1-e^{-kn/m})^k
\end{align*}
when $k$ and $m/n$ are constant.
Choosing $k$ appropriately (the optimal value for $k$ is $(m/n)\cdot \ln 2$), We see the false positive probability for an element falls exponentially with $m/n$, the number of bits per item used in the filter.

The idea of a learned Bloom filter is to train a neural network or other machine learning algorithm to recognize the set $S$.  We represent the algorithm by a function $f$, so that on input $x$ the algorithm returns a value $f(x)$ between 0 and 1. 
The algorithm ideally would return 1 for every element in the set and 0 for every element not in the set.  If we had such a predictor, we would not need any data structure, as we could just use the function to represent the set.  This is too much to expect in practice;  instead, we consider  an algorithm that returns a value $0 \leq f(x) \leq 1$.
We might intuitively interpret $f(x)$ as an estimate of the probability that $x$ is an element from the set, although this interpretation is not necessary in what follows.

We can choose a threshold $\tau$, and have the algorithm return that any element that satisfies $f(x) \geq \tau$ is in the set and otherwise it is not in the set.  Indeed, if we choose $\tau = \min_{x \in S} f(x)$ then there will be no false negatives.  But unless the predictor $f$ is very good, it is likely that this value of $\tau$ will lead to too many false positives.  

\begin{figure*}[t]
        \centering
        \includegraphics[width=0.8\textwidth]{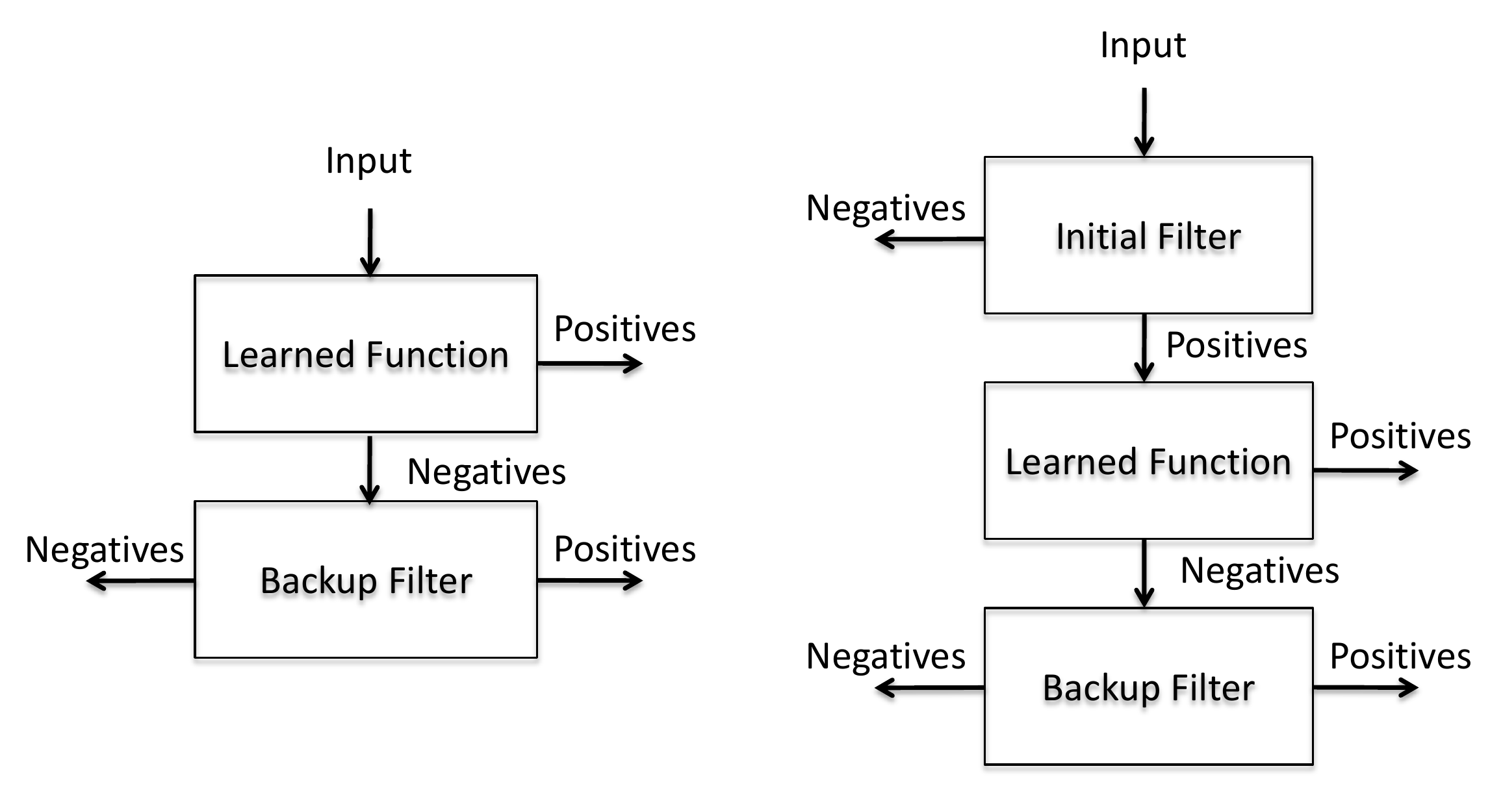}
                \caption{The left side shows a learned Bloom filter.  Negatives from the learned function are checked against the backup filter to prevent false negatives. The right side shows the sandwiched learned Bloom filter.  An initial filter removes many true negatives from reaching the learned function, reducing the false positives from the learned function.}
                \label{fig:diagram}
\end{figure*}

The alternative approach we apply is to use the learned function $f$ as a prefilter, selecting a larger value of $\tau$ to cut down on false positives, and then using a standard Bloom filter as a backup to prevent false negatives.  The setup is shown in Figure~\ref{fig:diagram}.  The initial learned function should correctly identify a substantial number of set elements, with a low false positive rate.  The backup Bloom filter then holds all the set elements that are incorrectly rejected by the learned function;  to be clear, we determine these in advance and set up the backup Bloom filter accordingly, which means the data set and the learned function must be fixed before setting up the back Bloom filter.  The backup Bloom filter again yields false positives, but prevents any false negatives.  

To see how there might be gains from this approach, imagine a small learned function that correctly identifies half of the original set.  Then the backup Bloom filter needs to only correct the erroneous false negatives of the predictor, which means the backup filter needs to represent only half the original set elements. Accordingly, the backup filter could be roughly half the size of a Bloom filter for the entire set with roughly the same false positive rate.  If the learned function has a small enough representation, namely less than half the size of a Bloom filter for the entire set, then this combination will be a win in terms of the space versus false positive probability tradeoff against a standard Bloom filter.  Empirical results from \citet{kraska2018case} show that learned Bloom filters can outperform standard Bloom filters for real-world data sets.    

We emphasize that the threshold $\tau$ will typically be chosen empirically, based on test queries, to predict the rate of false positives that will occur.  This empirical evaluation of test data to determine the relationship between $\tau$ and the rate of false positives we expect to find in future queries depends on our test queries being representative of the future;  otherwise we may obtain a higher false positive rate over future queries than expected.  A learned Bloom filter thus requires different additional assumptions than a standard Bloom filter in order to make statements about its performance.  Further details are discussed in \citet{mitzenmacher2018model}.  

Learned Bloom filters are relatively new;  given the large number of variations of Bloom filters, there may be interesting improvements for and variations of learned Bloom filters that will appear.  Indeed, it is already known that a ``sandwiched" learned Bloom filter that uses a learned filter between two standard Bloom filters, also shown in Figure~\ref{fig:diagram}, can yield better performance \citep{mitzenmacher2018model}.

\section{Caching with Predictions}
The caching or paging problem is both a canonical example of online algorithms, and a problem which has necessitated beyond worst-case analysis. 

Recall the problem setup. We are given a machine with a slow memory that can hold $N$ pages, and a fast memory with $k$ slots. Page requests arrive one at a time, and must be served out of fast memory. If the page is already in fast memory (cache), then a {\em hit} occurs, and the operation is free. Otherwise, there is a {\em cache miss}, the page is brought into fast memory, and one of the $k$ existing pages in the cache is evicted. The goal is to minimize the number of cache misses over the sequence of page requests. 

\subsection{What to Predict?}
The first question to address is to decide on the quantity that should be predicted by the machine learning subsystem. We look for predictions that are both useful to the algorithm and efficiently learnable.  The latter highlights the fact that predictions should be grounded in reality. Specifically, we want to make sure that we only need polynomially many examples to learn a good predictor; formally, we ensure that the function has a low sample complexity. As long as the family of functions specifying the predictor is relatively simple and well behaved, this condition is satisfied. However, an approach to fully predict the whole instance would fail the test and be untenable. 

What are good candidates for predictions for the paging problem? 
One algorithm that minimizes the number of cache misses is the Furthest-In-Future (FIF) algorithm, also known as B\'{e}l\'{a}dy's rule~\citep{belady1966study}. This method always evicts the element that is latest to come back. In order to be able to emulate it online, a useful prediction to be made at the time of each request is the {\em next} arrival time of this element. Formally, let $next(t)$ be the next arrival time of the element that appeared at time $t$, and $h(t)$ denote the predicted time of the next arrival of this element. 

Armed with such a predictor a natural approach is to plug it into the Furthest-In-Future algorithm, instead of the ground truth. We call this the PFIF for Predicted Furthest-In-Future. 

The analysis of the FIF algorithm directly implies that if the predictor $h$ is perfect, that is $h(t) = next(t)$ for all $t$, then PFIF is optimal. In other words, PFIF is consistent. But is the approach robust? 

First we must define an error metric. For a hypothesis $h$ let us define
$\eta(h) = \sum_t |h(t) - next(t)|.$ The question we want to ask is how the competitive ratio of PFIF scales with $\eta(h)$. The first thing to observe is that, as we defined it, the error grows with the input length. This is undesirable. Suppose we duplicate a request sequence and the predictions. The competitive ratio would remain the same, but the error defined above would double. We can normalize by the input length, but this, too, leads to pathological cases. For instance, take any request sequence of length $n$, and repeat the last element $n$ times. Since all of these extra requests would be cache hits, the performance of any algorithm remains the same as well. However, if the last $n$ predictions are perfect, then $\eta$ does not change, but error normalized by sequence length would decrease by a factor of two. Instead, we will normalize the error by the cost of the optimum solution \opt, which behaves correctly in both of these examples. 

We show that the competitive ratio of PFIF grows linearly with the error. Formally, the competitive ratio of PFIF is $\Omega(\eta(h)/\opt)$.

Consider a simple example with a cache of size 2, and three elements, $a$, $b$, and $c$. The true sequence will be $c,a,b,a,b,\ldots,a,b,c$. The predictions will be correct for elements $a$ and $b$, but the prediction for $c$ will always be at time $0$. Hence $\eta(h)$ is the length of the sequence.  In this case PFIF will keep $c$ in the cache, and suffer a cache miss almost every time. On the other hand, the optimal solution never misses on $a$ and $b$ once they are in the cache, and has a constant number of misses overall. We note that while it may be tempting to attempt to fix this algorithm by disregarding elements whose predicted appearance time has passed, this also has an $\Omega(\eta(h)/\opt)$ competitive ratio. 

\subsection{Marking Algorithms}
A natural question then is whether we can get competitive ratios with a more benign dependence on $\eta(h)/\opt$. 

To proceed we introduce the Marking family of algorithms, first introduced by~\citet{FiatKLMSY91}. These algorithms proceed in phases. Every phase begins with every cache position ``unmarked.'' Whenever there is a cache miss, an unmarked element is evicted, and the new element is marked. When a cache hit occurs the element is marked as well. This continues until all elements in the cache are marked, at which point the phase ends, and all of the marks are cleared. It is easy to show that any Marking algorithm is $O(k)$-competitive, where $k$ is the cache size.  Moreover,~\citet{FiatKLMSY91} show that if an algorithm evicts a uniformly random unmarked element, then the expected competitive ratio is $O(\log k)$. 

To prove a bound on the competitive ratio of the marking algorithm, we must get a lower bound on the optimum. To do so, we partition elements that arrive during a phase into two categories: clean and stale. {\em Clean} elements in phase $i$ are those that did not appear in phase $i-1$. In contrast, {\em stale} elements are those that were seen in the previous phase. Consider the following sequence with a cache of size $3$. 

$$\underbrace{a,a,b,a,b,c,}_{phase 1}\underbrace{b,b,c,b,d,}_{phase 2}\underbrace{a,a,d,c}_{phase 3}$$

\noindent Note that each phase ends as soon as three distinct elements appear. In phase 2, elements $b$ and $c$ are stale (since they appeared in phase 1), and element $d$ is clean. In contrast, in phase 3, $d$ is stale (as is $c$), and $a$ is clean. 

Let $C_i$ be the number of clean elements in phase $i$. Consider the performance of any algorithm on the clean elements. For some element $j \in C_i$, if it is not present in the cache in the beginning of phase $i$, then it will incur a cache miss. On the other hand, if it is in the cache at the beginning of the phase, it must have stayed in the cache throughout phase $i-1$, even though it did not appear, thus effectively reducing the working cache size. This argument can be made precise, to show that 
\begin{equation}
\label{eqn:clean}
OPT \geq \frac{1}{2} \sum_i C_i.
\end{equation}
In other words, the number of misses in any strategy is at least half the number of all clean elements.  We will relate the misses suffered by our algorithm to the number of clean elements in each phase. 

In order to utilize predictions in the marking framework, we modify the eviction strategy of the marking algorithm. If the arriving element is clean, we evict the unmarked element {\em predicted} to appear furthest in the future. If the arriving element is stale, we proceed as before, and evict a uniformly random unmarked element.  We refer to this variant 
as {\pmarker}.

\begin{theorem}
\pmarker\ has a competitive ratio of $O\big(\log \frac{\eta(h)}{\opt}\big)$.
\end{theorem}

To prove the theorem, let us try to understand the reason behind cache misses incurred by the algorithm. Suppose an element $e$ arrives and $e$ is not in the cache, causing a cache miss. If the element $e$ is clean, Equation \ref{eqn:clean} tells us we can charge its eviction directly to \opt. Suppose $e$ is stale. By the definition of stale elements, $e$ was in the cache when the phase began, thus it must have been evicted at some point between the beginning of the phase and its arrival. Let $ev(e)$ denote the element whose arrival caused the eviction of $e$. Either $ev(e)$ is clean, or it is another stale element, $e_1$, whose arrival time is earlier than $e$. In this case let us look why $e_1$ was evicted, i.e. $ev(e_1) = ev(ev(e))$. By the same logic, either $ev(e_1)$ is a clean element, or it is another stale element whose first arrival in this phase was earlier still. Therefore, repeatedly applying the $ev$ function to an element leads to a clean element whose arrival set off this chain of events. 

To get a bound on the competitive ratio, we ask how long can this chain be? This gives us the desired bound because each link in the chain represents a cache miss, each chain terminates with a clean element, and the number of clean elements is comparable to \opt\ by Equation \ref{eqn:clean}. It is clear that the length of the chain depends on the eviction rule: if we always evict the element that is latest to arrive (FIF) then each chain is of length $1$. If we do the reverse and evict the element that is next to arrive, then a chain can grow to be $\Omega(k)$ in length. 

We first analyze the standard Marking algorithm which evicts elements uniformly at random. 
\begin{lemma}
\label{lem:randommarker}
When evicting a random unmarked element, the expected length of each chain is $O(\log k)$. 
\end{lemma}

\begin{proof}
We need only consider stale elements in every phase, and there may be as many as $k-1$ of them. Order them by their arrival time, with $e_1$ arriving first, then $e_2$, and so on.  Denote by $L_i$ the length of the chain starting with element $e_i$. We can  write down the recurrence for $e_i$ as:
$$L_i = 1 + \frac{1}{k-i} \sum_{j = 1}^{k-1} L_j,$$
which solves to $L_0 = \Theta(\log k)$ when $L_{k-1} = 0$. 
\end{proof}

On the other hand, in \pmarker, when a clean element arrives, we evict the element predicted to arrive furthest in the future.  Let $c$ be a clean element that arrives at time $t_c$,  $s$ denote the element we chose to evict, and $t_s$ be the next time of arrival of $s$. Note that any stale element that arrives between $t_c$ and $t_s$ cannot increase the chain started at $c$. Therefore the only elements that can contribute to the growth of the chain are those who arrive after time $t_s$. But this is exactly in violation of our prediction, thus we can charge these cache misses to the error of the predictor. Let $inv_h(s)$ denote the set of elements that arrive after $s$ even though they were predicted by $h$ to arrive before. It is easy to extend Lemma \ref{lem:randommarker} to show that the length of the chain starting with $s$ is $\Theta(\log inv_h(s))$. 



To complete the analysis, we need to bound the number of inversions as a function of the accuracy of the predictor. 
For any two permutations, the total number of inversions and the $\ell_1$ distance of the elements are known to always be within a factor of two by the celebrated Diaconis-Graham inequality~\citep{DG}. The latter is also exactly $\eta(h)$ decomposed across phases. Further, since $\log$ is a concave function, to maximize the total length of all chains, we should partition errors equally among them. These two facts imply that the expected error of the above algorithm is $O(\log (\eta(h)/\opt))$. 

\subsection{Summary of Caching}
The caching problem is illustrative of the power of algorithms with predictions and the care that must be taken in designing them. We relied on the offline algorithm to identify the quantity that we wished to predict: the next appearance of every arriving element. We then proved that simply using this prediction as a proxy for the truth in the optimal offline algorithm allowed for pathological examples where the predictions led the algorithm astray. We then showed a different algorithm which, by using the predictions in a more careful manner, leads to a marked improvement in the competitive ratio over the na\"{i}ve way of using the predictor. In addition, we can show that even if the error is very large, we can guarantee performance within a constant factor of the standard marking algorithm. (See Exercise~\ref{exer:caching}.) 
Finally, as ~\citet{LykourisV18} showed, these gains are not just theoretical; even with off-the-shelf prediction models {\pmarker} consistently outperformed standard methods like the Least Recently Used (LRU) policy. 

\section{Scheduling with Predictions}

We consider the problem of scheduling jobs on a single machine to minimize the total flow time. One of the key points is that if job times are known, the simple greedy algorithm of Shortest Remaining Processing Time (SRPT) is optimal for this objective. Here we consider the potential of strategies such as SRPT in the context of scheduling and for queueing systems, where arrivals occur over time, but where the job times are only predicted, instead of known exactly.

\subsection{A Simple Model with Predictions}

We start with a very simple example.  Suppose we have $n$ jobs $j_1,\ldots,j_n$, each of which is either short or long.  Short jobs require time $s$ to process and long jobs require time $\ell > s$ to process.  Jobs are all available at time 0, and they are to be ordered and then
processed sequentially.  When the job times are known, shortest job first minimizes the total waiting time over all jobs.
If there are $n_s$ short
jobs and $n_\ell$ long jobs, it is easy to check that the average waiting time is
\begin{align*}
\frac{1}{n} \left ( n_s \frac{n_s -1}{2} s + n_\ell \frac{n_\ell -1}{2} \ell + n_\ell n_s s \right).
\end{align*}
If one has no information about the job times, then one
might randomly order the jobs, in which case the
expected waiting time over all jobs is
\begin{align*}
\frac{1}{n} \left ( n_s \left( \frac{n_s -1}{2} s + \frac{n_\ell}{2} \ell \right )
+ n_\ell \left( \frac{n_s}{2} s + \frac{n_\ell -1}{2} \ell \right ) \right ).
\end{align*}
Finally, suppose we have an algorithm that can predict a job's type.  We assume short jobs are misclassified as
long jobs with some probability $p$ and long jobs are misclassified as
short jobs with some probability $q$.  The natural approach would be to use
{\em shortest-predicted-job-first}; that is, we apply
shortest-job-first based on the predictions.  Some case arithmetic shows that the expected waiting time is then

{\scriptsize
  \setlength{\abovedisplayskip}{6pt}
  \setlength{\belowdisplayskip}{\abovedisplayskip}
  \setlength{\abovedisplayshortskip}{0pt}
  \setlength{\belowdisplayshortskip}{3pt}
\begin{align*}
& \frac{1}{n} \bigg(  (1-p) n_s \left( \frac{(1-p)(n_s -1)}{2} s + \frac{q n_\ell}{2} \ell \right) + pn_s \left( (1-p)(n_s -1) s + \frac{p(n_s -1)}{2} s
+ \frac{(1-q) n_\ell}{2}\ell + q n_\ell\ell \right ) \\
& +
(1-q) n_\ell \left( \frac{(1-q)(n_\ell -1)}{2} \ell + q (n_\ell-1) \ell
 + \frac{pn_s}{2} s + (1-p) n_s s \right) +
  q n_\ell \left( \frac{q(n_\ell-1)}{2} \ell
  + \frac{(1-p)n_s}{2} s \right) \bigg).
 \end{align*}
}%
With these expressions, one can determine the gain from using predictions over randomly ordering jobs, and the loss from using predictions in place of exact information. \citet{mitzenmacher2019scheduling}
suggests that we might also consider the {\em ratio}
between the expected waiting time with imperfect information and the expected
waiting time with perfect information.  \citet{mitzenmacher2019scheduling} further suggests that for any algorithm where it makes sense to use predicted information in
place of exact information one can consider this ratio, which is there referred to as the price of misprediction, using the following definition:
\begin{definition}
Let $M_A(Q;I)$ be the value of some measure (such as the expected
waiting time) for a system $Q$ given information $I$ about the system
using algorithm $A$, and let $M_A(Q;P)$ be the value of that metric
using predicted information $P$ in place of $I$ when using algorithm $A$.
Then the {\em price of misprediction} is defined as $M_A(Q;I)/M_A(Q;P)$.
\end{definition}
Notice here that (unlike many other uses of the ``price of'' language in algorithm analysis) the denominator is not necessarily an optimal algorithm, but the corresponding algorithm with exact information.  (One could, of course,
also compare against an optimal algorithm, as we have seen elsewhere in this chapter.) 

\subsection{More General Job Service Times}

We can consider a more general model where a job's actual and predicted time for service are real-valued random variables.  A natural probabilistic model is to suppose that the job sizes are governed by some distribution, and correspondingly, for each possible service time $x$, the output of the predictor $y$ is governed by some distribution that depends only on $x$.  For example, we might model the prediction $y$ as the value $x$ with some additional random noise, where the distribution of the noise might depend on $x$.  Equivalently, we can describe jobs according to a density function $g(x,y)$, giving the density for a job that has service time $x$ and predicted service time $y$. (For convenience we assume that $g(x,y)$ is ``well-behaved'' throughout, so that it is continuous and all necessary derivatives exist;  the analysis can be readily modified to handle point masses or other discontinuities in the distribution.)  This model makes some assumptions, most notably that each job corresponds to an independent instantiation of this density function.  However, it does seem sensible to model a machine learning algorithm that has been trained on lots of data as providing an estimated service time that corresponds to a conditional distribution based on the actual service time, as is done here, as long as the future jobs we are going to see can be thought of as coming from the same distribution as the jobs we used for training -- that is, roughly speaking, if the future is going to look like the past.

Again, we assume that all jobs are given at time 0, and we simply order the jobs according to the shortest predicted job first.  We let
$f_s(x) = \int_{y=0}^\infty g(x,y) \, dy$ be the corresponding density
function for the service time, and
$f_p(y) = \int_{x=0}^\infty g(x,y) \, dx$ be the corresponding density
function for the predicted service time.
If there are $n$ total jobs, the expected waiting time for a job using shortest job first given full information is given by
\begin{align*}
(n-1) \int_{x=0}^\infty f_s(x) \left (\int_{z=0}^x z f_s(z) \,dz \right ) dx,
\end{align*}
while the expected waiting time for a job using predicted information using shortest predicted job first is given by
\begin{align*}
(n-1) \int_{y=0}^\infty f_p(y) \left (\int_{x=0}^\infty \int_{z=0}^y x g(x,z) \, dz \, dx \right ) dy.
\end{align*}
In words, in the full information case, given the service time for a job, we determine its expected waiting time from each other job by taking the expectation conditioned on the other job having a smaller service time.  In the predicted information case, to compute the expected
waiting time for a job given its predicted service time, 
we determine its expected waiting time from each other job by taking the expectation based on the other job's actual service time, conditioned on the other job having a smaller predicted service time than the original job.

In this case, the price of misprediction is given by the ratio
\begin{align}
\label{eq:pom1}
\frac{\int_{y=0}^\infty f_p(y) \left (\int_{x=0}^\infty \int_{z=0}^y x g(x,z) \, dz \, dx \right ) dy}{\int_{x=0}^\infty f_s(x) \left (\int_{z=0}^x z f_s(z) \,dz \right ) dx};
\end{align}
while this is not the simplest of expressions, given $g(x,y)$ it can be numerically evaluated.  As an interesting albeit not necessarily realistic example, suppose that jobs have service times that are exponentially distributed with mean 1, but the service time prediction for a job with actual service time $x$ is exponential with mean $x$, so that the mean of the prediction is correct but the prediction itself can be significantly inaccurate.  It can be shown that the price of misprediction in this case is $4/3$; this is given as Exercise~\ref{exer:exponential}.

\subsection{Scheduling Queues}

This type of analysis can be extended, with some more involved work, to the case of queues.  In the queueing setting, we still just have one machine, but jobs both enter for service and leave after finishing service over time, and we typically first look at the average time in the system when considering performance.  For example, in standard queueing theory, the prototypical queue is known as the {M/M/1} queue, where arrivals are a Poisson process of rate $\lambda < 1$, service times  are independently and identically exponentially distributed with mean 1, and there is a single server serving the customers.  (The``M'' in the {M/M/1} queue stands for memoryless.)  One of the fundamental results in queueing theory is that the expected time a customer spends  waiting for and obtaining service in equilibrium  in an  {M/M/1} queue with First Come First Served (FCFS) scheduling (also called First In First Out (FIFO)) is given by ${\nicefrac{1}{(1-\lambda)}}$.  In this section we consider queues with Poisson arrivals but general service time distributions, not just exponential.  

If one knows the service time for a job, one can try to schedule better than FCFS.  Shortest Job First (SJF) is the non-preemptive strategy that schedules the queued job with the shortest service time when a job completes.  Preemptive Shortest Job First (PSJF) acts similarly, but will preempt a running job if new job with a smaller service time arrives.  Shortest Remaining Processing Time (SRPT) will instead schedule and preempt jobs based on their remaining processing time instead of their service time.  

In \citet{mitzenmacher2019scheduling}, these strategies are considered in the setting where one has predicted service times instead of actual service times, leading to the strategies
Shortest Predicted Job First (SPJF), Preemptive Shortest Predicted Job First (PSPJF), and Shortest Predicted Remaining Processing Time (SPRPT).  Equations for all three strategies are provided under the assumption that there is a joint density distribution $g(x,y)$ for jobs with service time $x$ and predicted service time $y$, and that each job independently yields predicted and actual service times from this distribution.

For example, comparing SJF and SPJF, we first set up the following notation.
Let $f_s(x) = \int_{y=0}^\infty g(x,y) \, dy$ and
$f_p(y) = \int_{x=0}^\infty g(x,y) \, dx$ be the corresponding 
service and predicted service density functions.  
Finally, the quantity $\rho_x = \lambda \int_{t=0}^x tf_s(t) \, dt$ is the rate of work entering the queue from jobs with service time at most $x$,
and $\rho'_y = \lambda \int_{t=0}^y \int_{x=0}^\infty g(x,t) x \, dx \, dt$ is the corresponding rate of work entering the queue from jobs with predicted service at most $y$.  

For SJF, it is known that $W(x)$, the time spent waiting in the queue (not being served)
for jobs with service time
$x$, in the steady state satisfies
\begin{align*}
\E[W(x)] = \frac{\rho \E[S^2]}{2\E[S] \left ( 1 - \rho_x \right )^2}.
\end{align*}
Note that the waiting time for a job with service time $x$ depends on the general service distribution but also specifically on the work from jobs with service time at most $x$, as one might expect.   
The overall expected time waiting in a queue, which
we denote by $\E[W]$, is then simply
\begin{align*}
\E[W] = \int_{x=0}^\infty f(x) \E[W(x)] \, dx.
\end{align*}
It turns out that for SPJF, a similar analysis to that used to derive the performance equations for SJF applies.
If we let
$W'(y)$ be the distribution of time spent waiting in the queue for a job with predicted service time $y$
in the steady state,
then
\begin{align*}
\E[W'(y)] = \frac{\rho \E[S^2]}{2\E[S] \left ( 1 - \rho'_y \right )^2}.
\end{align*}
The price of misprediction for the time waiting in queue for SJF/SPJF is then expressed as 
\begin{align*}
\frac{\int_{y=0}^\infty \frac{f_p(y)}{(1-\rho'_y)^2} \, dy}
{\int_{x=0}^\infty \frac{f_s(x)}{(1-\rho_x)^2} \, dx} .
\end{align*}

Similar analyses can be done for PSJF/PSPJF and SRPT/SPRPT, although the resulting expressions are more complicated.  

Simulations show that even fairly weak predictors can provide significant performance gains for queues under high load (that is, as $\lambda$ gets close to 1), as FIFO queues relatively frequently stack short jobs behind a long job, and this is a primary reason for long expected waiting times.  Predictors that simply keep long jobs behind short jobs most of the time therefore greatly improve the expected waiting time over all jobs.  For example, a predictor with a multiplicative error can do quite well.  
Figure~\ref{fig:figb} provides an example with $\lambda = 0.95$ and two types of service distributions: exponential with mean 1, and a Weibull distribution with cumulative distribution $1-e^{-\sqrt{2x}}$.  (The
Weibull distribution also has mean 1, but is more heavy-tailed, so longer jobs occur with higher probability.)  The results are averaged over 1000 trials over a time period of 1 million time units, where each trial averages the time in system for jobs that complete after the first $100,000$ time units. A job with service time $x$ has a predicted service
time that is uniform over $[(1-\alpha)x,(1+\alpha)x]$ for a parameter $\alpha$;  we try $\alpha = j/10$ for integer $j$ from 0 to 9.  We observe that performance degrades gracefully with $\alpha$, and is much better than without predictions, where the steady state average time in the system is 20 for the exponential distribution and 58 for the Weibull distribution.

\begin{figure}
    \centering
    \includegraphics[width=4.0in]{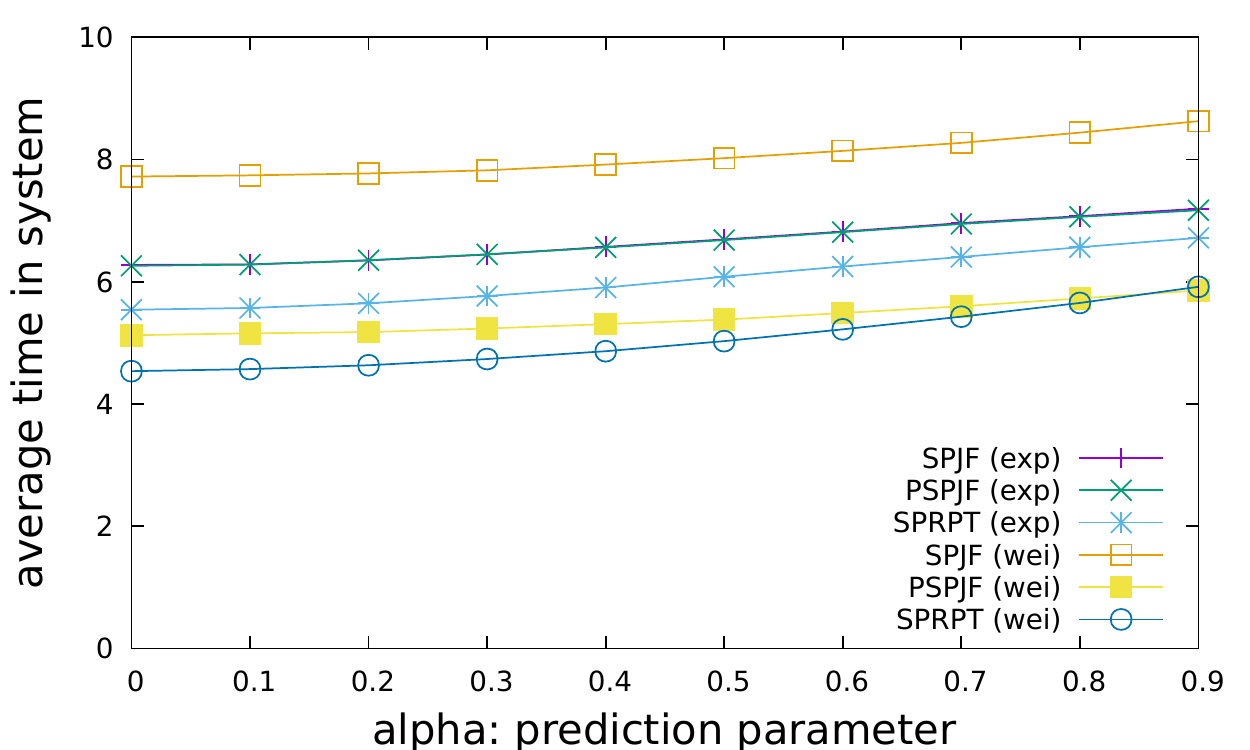}
    \caption{Results from simulations at $\lambda = 0.95$ for exponential and Weibull distributions.
A job with service time $x$ has predicted service time uniform over $[(1-\alpha)x,(1+\alpha)x]$.
Performance degrades gracefully with $\alpha$.}
    \label{fig:figb}
\end{figure}

\section{Bibliographic Notes}

The use of advice to assist online algorithms has been studied in the past \citep{boyar2016online}.  But previous work has focused on minimizing the number of advice bits from omniscient sources to achieve optimal or near-optimal competitive ratios.  The motivation of the work in online algorithms using learning-based predictions more closely mirrors the use of machine learning in practice, focusing on improvements in the competitive ratio that can arise with realistic advice.   

The idea of learning in order to improve algorithms' performance, especially in the realm of online algorithms, has appeared in some works in the past. For instance,~\citet{DevenurHayes} and \citet{VeeSV} explored how predictions can be used to obtain nearly optimal online matching bounds, while~\citet{ColeR} and \citet{MedinaV17} showed how to learn from samples to maximize revenue in auction settings. In parallel~\citet{kraska2018case} showed that these endeavors are not simply theoretical, building a system that used machine learning to improve retrieval speed for index data structures.   

A formal model of learning with predictions, including the notions of $\alpha$-consistency and $\beta$-robustness, was presented by~\citet{LykourisV18}. They were also the first ones to analyze this setting for the caching problem. The analysis we presented here is due to \citet{Rohatgi}. Additionally~\citet{purohit2018improving} demonstrated explicit trade-offs between these two concepts in the context of ski-rental and online scheduling. 

A good general reference for queueing theory, including derivations for SJF and SRPT with exact information, is 
\citet{harchol2013performance}.  

In scheduling for queues, some works have looked at the effects of using imprecise information for load balancing in multiple queue settings.  For example, \citet{mitzenmacher2000useful} considers using old load information to place jobs in the context of the power of two choices.
For single queues, \citet{wierman2008scheduling} look at variations of SRPT and SJF with inexact job sizes, bounding the performance gap based on bounds on how inexact the estimates can be. Dell'Amico, Carra, and Michardi
empirically study scheduling policies for
queueing systems with estimated sizes \citep{dell2015psbs}.
As mentioned, \cite{purohit2018improving}
specifically looked at scheduling with predictions in the standard online setting, where they considered variants of shortest predicted processing time that
yield good performance in terms of the competitive ratio, with the performance
depending on the accuracy of the predictions.

The Count-Min Sketch \citep{cormode2005improved} and the Count-Sketch \citep{charikar2002finding} are well known data structures for finding heavy hitters in data streams, and have found many additional applications.  

Bloom filters were originally developed by \citet{bloom1970space}, and have proven useful for a number of applications \citep{broder2004network}.  Learned Bloom filters were originally described by \citet{kraska2018case}, where other additional possible examples of using learning to improve index data structures were proposed.

\bibliography{chapter}
\bibliographystyle{cambridgeauthordate}


\section*{Exercises}
\begin{enumerate}
\item \label{exer:skirentla}
Prove the competitive ratio bound given
in equation~\ref{eqn:skirental}
for the ski rental with predictions algorithm.   
\item \label{exer:caching}
Consider the caching problem, and suppose we have two data-dependent eviction algorithms. For an input $x$, one of them has competitive ratio $a(x)$ while the other has ratio $b(x)$. Develop an algorithm that for every input $x$ has competitive ratio $O(\min(a(x), b(x)))$. 
\item \label{exer:exponential} 
Consider the setting of equation~\ref{eq:pom1}, where job sizes are exponentially distributed with mean 1, and a job with mean service time $x$ has a predicted service time that is itself exponentially distributed with mean $x$.  Show via numerical evaluation or integration (perhaps using a software package for evaluating integrals) that the ``price of misprediction'' in this case is $4/3$.
\item \label{exer:simulation}
Write a simulation to study one of the problems discussed in the chapter.  For example, you could write a simulation for a queue that uses predicted service times, and use it to explore how the service time distribution and the quality of the prediction affect the average time spent waiting in the queue.  Or you could implement an Count-Min sketch and simulate a predictor for heavy hitter elements, and use it to explore how the accuracy of the sketch improves with the quality of the prediction or varies with how skewed the frequency distribution of items is.  Your simulation can use an actual learned function as a predictor, or you could use a synthetic prediction (by for example adding noise in some specified way to the ground truth to obtain a prediction).  
\end{enumerate}

\end{document}